\documentclass[11pt]{article}
\usepackage{hyperref}
\usepackage{times}  
\usepackage{mathpazo}
\usepackage{amssymb,amsmath,amsthm}
\usepackage{epsfig}
\usepackage{amsmath}
\usepackage{fullpage}

\def\lsoft{{l\kern-0.035cm\char39\kern-0.03truecm}}

\newtheorem{theorem}{Theorem}[section]
\newtheorem{proposition}[theorem]{Proposition}
\newtheorem{definition}[theorem]{Definition}

\newtheorem{lemma}[theorem]{Lemma}

\newtheorem{corollary}[theorem]{Corollary}

\newcommand{\qedsymb}{\hfill{\rule{2mm}{2mm}}}
\renewenvironment{proof}[1][]{\begin{trivlist}
\item[\hspace{\labelsep}{\bf\noindent Proof#1:\/}] }{\qedsymb\end{trivlist}}

\def\calF{{\cal F}}

\def\calS{{\cal S}}

\def\R{\mathbb{R}}

\newcommand{\Uncovered}{\textsc{Unfair-IS-Cycle}}

\newcommand{\SchrijverP}{\textsc{Schrijver}}
\newcommand{\KneserP}{\textsc{Kneser}}

\newcommand\Expec[2]{{{\bf E}_{#1}\left[ {#2} \right]}}

\newcommand{\NP}{\mathsf{NP}}

\newcommand{\PPA}{\mathsf{PPA}}
\newcommand{\TFNP}{\mathsf{TFNP}}
\newcommand{\LEAF}{\textsc{Leaf}}

\newcommand{\FISC}{\textsc{Fair-IS-Cycle}}

\newcommand{\CT}{\textsc{Cycle-Plus-Triangles}}

\renewcommand{\epsilon}{\varepsilon}

\newcommand{\Image}{\mathop{\mathrm{Im}}}

\newcommand{\Sset}{{\mathbb S}}

\begin{document}

\title{{\bf On Finding Constrained Independent Sets in Cycles}}
\author{
Ishay Haviv\thanks{School of Computer Science, The Academic College of Tel Aviv-Yaffo, Tel Aviv 61083, Israel. Research supported in part by the Israel Science Foundation (grant No.~1218/20).}
}

\date{}

\maketitle

\begin{abstract}
A subset of $[n] = \{1,2,\ldots,n\}$ is called stable if it forms an independent set in the cycle on the vertex set $[n]$.
In 1978, Schrijver proved via a topological argument that for all integers $n$ and $k$ with $n \geq 2k$, the family of stable $k$-subsets of $[n]$ cannot be covered by $n-2k+1$ intersecting families.
We study two total search problems whose totality relies on this result.

In the first problem, denoted by $\SchrijverP(n,k,m)$, we are given an access to a coloring of the stable $k$-subsets of $[n]$ with $m = m(n,k)$ colors, where $m \leq n-2k+1$, and the goal is to find a pair of disjoint subsets that are assigned the same color. While for $m = n-2k+1$ the problem is known to be $\PPA$-complete, we prove that for $m < d \cdot \lfloor \frac{n}{2k+d-2} \rfloor$, with $d$ being any fixed constant, the problem admits an efficient algorithm.
For $m = \lfloor n/2 \rfloor-2k+1$, we prove that the problem is efficiently reducible to the $\KneserP$ problem. Motivated by the relation between the problems, we investigate the family of {\em unstable} $k$-subsets of $[n]$, which might be of independent interest.

In the second problem, called Unfair Independent Set in Cycle, we are given $\ell$ subsets $V_1, \ldots, V_\ell$ of $[n]$, where $\ell \leq n-2k+1$ and $|V_i| \geq 2$ for all $i \in [\ell]$, and the goal is to find a stable $k$-subset $S$ of $[n]$ satisfying the constraints $|S \cap V_i| \leq |V_i|/2$ for $i \in [\ell]$.
We prove that the problem is $\PPA$-complete and that its restriction to instances with $n=3k$ is at least as hard as the Cycle plus Triangles problem, for which no efficient algorithm is known. On the contrary, we prove that there exists a constant $c$ for which the restriction of the problem to instances with $n \geq c \cdot k$ can be solved in polynomial time.
\end{abstract}

\section{Introduction}

For integers $n$ and $k$ with $n \geq 2k$, the Kneser graph $K(n,k)$ is the graph whose vertices are all the $k$-subsets of $[n]= \{1,2\ldots,n\}$, where two such sets are adjacent in the graph if they are disjoint. The graph $K(n,k)$ admits a proper vertex coloring with $n-2k+2$ colors. This indeed follows by assigning the color $i$, for each $i \in [n-2k+1]$, to all the vertices whose minimal element is $i$, and the color $n-2k+2$ to the remaining vertices, those contained in $[n] \setminus [n-2k+1]$. In 1978, Lov{\'{a}}sz~\cite{LovaszKneser} proved, settling a conjecture of Kneser~\cite{Kneser55}, that fewer colors do not suffice, that is, the chromatic number of the graph satisfies $\chi(K(n,k)) = n-2k+2$.
Soon later, Schrijver~\cite{SchrijverKneser78} strengthened Lov{\'{a}}sz's result by proving that the subgraph $S(n,k)$ of $K(n,k)$ induced by the stable $k$-subsets of $[n]$, i.e., the vertices of $K(n,k)$ that form independent sets in the cycle on the vertex set $[n]$, has the same chromatic number. It was further shown in~\cite{SchrijverKneser78} that the graph $S(n,k)$ is vertex-critical, in the sense that any removal of a vertex from the graph decreases its chromatic number.

It is interesting to mention that despite the combinatorial nature of Kneser's conjecture~\cite{Kneser55}, Lov{\'{a}}sz's proof~\cite{LovaszKneser} relies on the Borsuk--Ulam theorem~\cite{Borsuk33}, a fundamental result in the area of algebraic topology. Several alternative proofs and extensions were provided in the literature over the years (see, e.g.,~\cite{MatousekBook,MatousekZ04}). Although they are substantially different from each other, they all essentially rely on topological tools.

The computational search problem associated with Kneser graphs, denoted by $\KneserP$, was proposed by Deng, Feng, and Kulkarni~\cite{DengFK17} and is defined as follows.
Its input consists of integers $n$ and $k$ with $n \geq 2k$ and an access to a coloring of the vertices of $K(n,k)$ with $n-2k+1$ colors. The goal is to find a monochromatic edge in the graph, i.e., two disjoint $k$-subsets of $[n]$ that are assigned the same color by the given coloring. Since the number of colors used by the input coloring is strictly smaller than the chromatic number of $K(n,k)$~\cite{LovaszKneser}, it follows that this search problem is total, in the sense that every input is guaranteed to have a solution. Note that the input coloring may be given as an oracle access that provides the color of any queried vertex, and that an algorithm for the problem is considered efficient if its running time is polynomial in $n$. In other variants of the problem, the input coloring is given by some succinct representation, e.g., a Boolean circuit or an efficient Turing machine. The computational search problem $\SchrijverP$ is defined similarly, where the input represents a coloring of the vertices of $S(n,k)$ with $n-2k+1$ colors, and the goal is to find a monochromatic edge, whose existence is guaranteed by the aforementioned result of Schrijver~\cite{SchrijverKneser78}.

The computational complexity of the $\SchrijverP$ problem was determined in~\cite{Haviv22-FISC}, where it was shown to be complete in the complexity class $\PPA$.
This complexity class, introduced in 1994 by Papadimitriou~\cite{Papa94}, is known to capture the complexity of several additional total search problems whose totality is based on the Borsuk--Ulam theorem, e.g., Consensus Halving, Bisecting Sandwiches, and Splitting Necklaces~\cite{FG19}.
Note that this line of $\PPA$-completeness results is motivated not only from the computational complexity perspective, but also from a mathematical point of view, as one may find those results as an indication for the necessity of topological arguments in the existence proof of the solutions of these problems.
As for the $\KneserP$ problem, it is an open question whether it is also $\PPA$-hard, as was suggested by Deng et al.~\cite{DengFK17}.
We remark that its complexity is related to that of the Agreeable Set problem from the area of resource allocation (see~\cite{ManurangsiS19,Haviv22a}).
The $\KneserP$ and $\SchrijverP$ problems were also investigated in the framework of parameterized algorithms~\cite{Haviv22a,Haviv22b}, where it was shown that they admit randomized fixed-parameter algorithms with respect to the parameter $k$, namely, algorithms whose running time is $n^{O(1)} \cdot k^{O(k)}$ on input colorings of $K(n,k)$ and $S(n,k)$.

Before turning to our results, let us mention another computational search problem, referred to as the $\CT$ problem.
Its input consists of an integer $k$ and a graph on $3k$ vertices, whose edge set is the disjoint union of a Hamilton cycle and $k$ pairwise vertex-disjoint triangles. The goal is to find an independent set of size $k$ in the graph. The existence of a solution for every input of the problem follows from a result of Fleischner and Stiebitz~\cite{FleischnerS92}, which settled in the early nineties a conjecture of Du, Hsu, and Hwang~\cite{DHH93} as well as its strengthening by Erd{\"{o}}s~\cite{ErdosQuestions90}. Their proof in fact shows that every such graph is $3$-choosable, and thus $3$-colorable, so in particular, it contains an independent set of size $k$. Here, however, the existence of a solution for every input of the problem is known to follow from several different arguments. While the proof of~\cite{FleischnerS92} relies on the polynomial method in combinatorics (see also~\cite{AlonNull99}), an elementary proof was given slightly later by Sachs~\cite{Sachs93}, and another proof, based on the chromatic number of $S(n,k)$, was provided quite recently by Aharoni et al.~\cite{AABCKLZ17}. Yet, none of these proofs is constructive, in the sense that they do not suggest an efficient algorithm for the $\CT$ problem. The question of whether the problem admits an efficient algorithm was asked by several authors and is still open (see, e.g.,~\cite{FleischnerS97,AABCKLZ17,AlonFair22}).
Interestingly, the approach of~\cite{AABCKLZ17} implies that the problem is not harder than the restriction of the $\SchrijverP$ problem to colorings of $S(n,k)$ with $n=3k$.

\subsection{Our Contribution}

In this paper, we introduce two total search problems concerned with finding stable sets under certain constraints. The totality of the problems relies on the chromatic number of the graph $S(n,k)$~\cite{SchrijverKneser78}. We study these problems from algorithmic and computational perspectives. In what follows, we describe the two problems and our results on each of them.

\subsubsection{The Generalized Schrijver Problem}

We start by considering a generalized version of the $\SchrijverP$ problem, which allows the number of colors used by the input coloring to be any prescribed number.
Let $\SchrijverP(n,k,m)$ denote the problem which asks to find a monochromatic edge in $S(n,k)$ for an input coloring that uses $m = m(n,k)$ colors. Note that every input of the problem is guaranteed to have a solution whenever $m \leq n-2k+1$, and that for $m=n-2k+1$, the problem coincides with the standard $\SchrijverP$ problem.

The $\SchrijverP(n,k,m)$ problem obviously becomes easier as the number of colors $m$ decreases.
For example, it is not difficult to see that for $m = \lfloor n/k \rfloor-1$, the problem can be solved efficiently, in time polynomial in $n$.
Indeed, the clique number of the graph $S(n,k)$ is $\lfloor n/k \rfloor$, which is strictly larger than $m$, so by querying the input coloring for the colors of the vertices of a clique of maximum size, one can find two adjacent vertices with the same color. Our first result extends this observation and essentially shows that the $\SchrijverP(n,k,m)$ problem can be solved efficiently for any number of colors $m$ satisfying $m = O(n/k)$.

\begin{theorem}\label{thm:Sch_m_alpha}
For every integer $d \geq 2$, there exists an algorithm for the $\SchrijverP(n,k,m)$ problem with $m < d \cdot \lfloor \frac{n}{2k+d-2} \rfloor$ whose running time is $n^{O(d)}$.
\end{theorem}

Our next result relates the generalized $\SchrijverP(n,k,m)$ problem to the $\KneserP$ problem.
\begin{theorem}\label{thm:Schr<Kneser}
$\SchrijverP(n,k,\lfloor n/2 \rfloor-2k+1)$ is polynomial-time reducible to $\KneserP$.
\end{theorem}

The simple proof of Theorem~\ref{thm:Schr<Kneser} involves a proper coloring of the subgraph of $K(n,k)$ induced by the {\em unstable} $k$-subsets of $[n]$, i.e., the vertices of $K(n,k)$ that do not form vertices of $S(n,k)$. This graph, which we denote by $U(n,k)$, can be properly colored using $\lceil n/2 \rceil$ colors. Indeed, every unstable $k$-subset of $[n]$ includes an odd element, hence by assigning to each vertex of $U(n,k)$ some odd element that belongs to its set, we obtain a proper coloring of the graph with the desired number of colors. Since $U(n,k)$ is a subgraph of $K(n,k)$, it follows that for all admissible values of $n$ and $k$, we have $\chi(U(n,k)) \leq \min (n-2k+2, \lceil n/2 \rceil )$.

Motivated by the reduction given by Theorem~\ref{thm:Schr<Kneser}, we further explore the graph $U(n,k)$, whose study may be of independent interest.
We prove that the above upper bound on the chromatic number is essentially tight (up to an additive $1$ in certain cases; see Corollary~\ref{cor:chiU(n,k)} and the discussion that follows it).
The proof is topological and applies the Borsuk--Ulam theorem. We further determine the independence number of the graph $U(n,k)$ (see Theorem~\ref{thm:alphaU(n,k)}), using a structural result of Hilton and Milner~\cite{HM67} on the largest non-trivial intersecting families of $k$-subsets of $[n]$.

The motivation for Theorem~\ref{thm:Schr<Kneser} comes from the fact that the $\SchrijverP$ problem is known to be $\PPA$-hard, whereas no hardness result is known for the $\KneserP$ problem.
However, in a subsequent work we show that under plausible complexity assumptions it is unlikely that the $\SchrijverP(n,k,m)$ problem with $m = \lfloor n/2 \rfloor-2k+1$ is $\PPA$-hard.
Yet, it is of interest to figure out whether or not the problem admits an efficient algorithm.
While this challenge is left open, the following result shows that the problem is not harder than the restriction of the standard $\SchrijverP$ problem to colorings of $S(n,k)$ with $n=4k$.

\begin{theorem}\label{thm:Sch_4k}
If there exists a polynomial-time algorithm for the restriction of the $\SchrijverP$ problem to colorings of $S(n,k)$ with $n=4k$, then there exists a polynomial-time algorithm for the $\SchrijverP(n,k,m)$ problem where $m = \lfloor n/2 \rfloor -2k+1$.
\end{theorem}

We finally observe that the restriction of $\SchrijverP(n,k,m)$ with $m = \lfloor n/2 \rfloor-2k+1$ to instances satisfying $n = \Omega(k^4)$ admits an efficient randomized algorithm. This essentially follows from the fixed-parameter algorithm presented in~\cite{Haviv22b} (see Section~\ref{sec:Sch} for details).

\subsubsection{The Unfair Independent Set in Cycle Problem}

The second problem studied in this paper is the Unfair Independent Set in Cycle problem, denoted by $\Uncovered$ and defined as follows.
Its input consists of two integers $n$ and $k$ with $n \geq 2k$ and $\ell$ subsets $V_1, \ldots, V_{\ell}$ of $[n]$, where $\ell \leq n-2k+1$ and $|V_i| \geq 2$ for all $i \in [\ell]$. The goal is to find a stable $k$-subset $S$ of $[n]$ that satisfies the constraints $|S \cap V_i| \leq |V_i|/2$ for $i \in [\ell]$.
The name of the problem essentially borrows the terminology of~\cite{AABCKLZ17}, where a set is said to fairly represent a set $V_i$ if it includes at least roughly half of its elements, hence the desired stable set in the $\Uncovered$ problem is required to {\em unfairly} represent each of the given sets $V_i$.
It is not difficult to show, using the chromatic number of $S(n,k)$, that every input of the $\Uncovered$ problem has a solution (see Lemma~\ref{lemma:Unco_total}).
Note that the requirement that the input sets satisfy $|V_i| \geq 2$ for all $i \in [\ell]$ is discussed in Section~\ref{sec:comput_prob}.

It is natural to compare the definition of the $\Uncovered$ problem to that of the Fair Independent Set in Cycle problem, denoted by $\FISC$ and studied in~\cite{Haviv22-FISC} (see Definition~\ref{def:FISC}).
While the goal in the former is to find a stable subset of $[n]$ with a prescribed size $k$ that includes {\em no more} than half of the elements of each $V_i$, the goal in the latter is, roughly speaking, to find a stable subset of $[n]$, of an arbitrary size, that includes {\em at least} half of the elements of each $V_i$.
The specification of the size $k$ in the inputs of $\Uncovered$ makes the problem non-trivial and allows us to study it for various settings of the quantities $n$ and $k$.

The following result shows that the complexity of the $\Uncovered$ problem is perfectly captured by the class $\PPA$.
This is established using the $\SchrijverP$ and $\FISC$ problems which are $\PPA$-complete~\cite{Haviv22-FISC}.

\begin{theorem}\label{thm:Unc_PPA}
The $\Uncovered$ problem is $\PPA$-complete.
\end{theorem}

We next consider some restrictions of the $\Uncovered$ problem to instances in which the integer $n$ is somewhat larger than $2k$.
On the one hand, the restriction of the problem to instances with $n=3k$ is at least as hard as the $\CT$ problem, for which no efficient algorithm is known (see Proposition~\ref{prop:CT<Un}). On the other hand, we prove that on instances whose ratio between $n$ and $k$ is above some absolute constant, the problem can be solved in polynomial time.

\begin{theorem}\label{thm:alg_c}
There exists a constant $c >0$, such that there exists a polynomial-time algorithm for the restriction of the $\Uncovered$ problem to instances with $n \geq c \cdot k$.
\end{theorem}
\noindent
The proof of Theorem~\ref{thm:alg_c} is based on a probabilistic argument with alterations, which is derandomized into a deterministic algorithm using the method of conditional expectations (see, e.g.,~\cite[Chapters~3~and~16.1]{AlonS16}).
The approach is inspired by a probabilistic argument of Kiselev and Kupavskii~\cite{KiselevK20}, who proved that for $n \geq (2+o(1)) \cdot k^2$, every proper coloring of the Kneser graph $K(n,k)$ with $n-2k+2$ colors has a trivial color class (all of whose members share a common element).

\subsection{Outline}
The rest of the paper is organized as follows.
In Section~\ref{sec:preliminaries}, we collect some definitions and results that will be used throughout the paper.
In Section~\ref{sec:Sch}, we study the generalized $\SchrijverP$ problem and prove Theorems~\ref{thm:Sch_m_alpha},~\ref{thm:Schr<Kneser}, and~\ref{thm:Sch_4k}.
In Section~\ref{sec:Constrained}, we study the $\Uncovered$ problem and prove Theorems~\ref{thm:Unc_PPA} and~\ref{thm:alg_c}.
Finally, in Section~\ref{sec:unstable}, we consider the family of unstable $k$-subsets of $[n]$ and study the chromatic and independence numbers of the graph $U(n,k)$.

\section{Preliminaries}\label{sec:preliminaries}

\subsection{Kneser and Schrijver Graphs}\label{sec:KS_graphs}

For integers $n$ and $k$, let $\binom{[n]}{k}$ denote the family of all $k$-subsets of $[n]$.
A subset of $[n]$ is called {\em stable} if it does not include two consecutive elements nor both $1$ and $n$, equivalently, it forms an independent set in the cycle on the vertex set $[n]$ with the natural order along the cycle. Otherwise, the set is called {\em unstable}. The family of stable $k$-subsets of $[n]$ is denoted by $\binom{[n]}{k}_{\mathrm{stab}}$. The Kneser graph and the Schrijver graph are defined as follows.

\begin{definition}
For integers $n$ and $k$ with $n \geq 2k$, the {\em Kneser graph} $K(n,k)$ is the graph on the vertex set $\binom{[n]}{k}$, where two sets $A,B \in \binom{[n]}{k}$ are adjacent if they satisfy $A \cap B = \emptyset$. The {\em Schrijver graph} $S(n,k)$ is the subgraph of $K(n,k)$ induced by the vertices of $\binom{[n]}{k}_{\mathrm{stab}}$.
\end{definition}

Obviously, the number of vertices in $K(n,k)$ is $\binom{n}{k}$. The number of vertices in $S(n,k)$ is given by the following lemma (see, e.g.,~\cite[Fact~4.1]{Haviv22a}).
\begin{lemma}\label{lemma:|S(n,k)|}
For all integers $n$ and $k$ with $n \geq 2k$, the number of stable $k$-subsets of $[n]$ is
$\frac{n}{k} \cdot \binom{n-k-1}{k-1}$.
\end{lemma}
\noindent
With Lemma~\ref{lemma:|S(n,k)|} at hand, one can derive the following (see also~\cite[Lemma~1]{Talbot03}).
\begin{lemma}\label{lemma:Talbot}
For all integers $n$ and $k$ with $n \geq 2k$ and for every $i \in [n]$, the number of stable $k$-subsets of $[n]$ that include $i$ is
$\binom{n-k-1}{k-1}$.
\end{lemma}

As usual, we denote the independence number of a graph $G$ by $\alpha(G)$, and its chromatic number by $\chi(G)$.
The chromatic numbers of $K(n,k)$ and $S(n,k)$ were determined, respectively, by Lov{\'{a}}sz~\cite{LovaszKneser} and by Schrijver~\cite{SchrijverKneser78}, as stated below.
\begin{theorem}[{\cite{LovaszKneser,SchrijverKneser78}}]\label{thm:chi_KorS(n,k)}
For all integers $n$ and $k$ with $n \geq 2k$, $\chi(K(n,k)) = \chi(S(n,k)) = n-2k+2$.
\end{theorem}

\subsection{Intersecting Families}

A family $\calF$ of sets is called {\em intersecting} if for every two sets $A,B \in \calF$ it holds that $A \cap B \neq \emptyset$.
Note that a family of $k$-subsets of $[n]$ is intersecting if and only if it forms an independent set in the graph $K(n,k)$.
An intersecting family $\calF$ is said to be {\em trivial} if there exists an element that belongs to all members of $\calF$. Otherwise, the family $\calF$ is {\em non-trivial}.
The famous Erd{\"{o}}s-Ko-Rado theorem~\cite{EKR61} asserts that the largest size of an intersecting family of $k$-subsets of $[n]$ is $\binom{n-1}{k-1}$, which is attained by the maximal trivial intersecting families.
The following result of Hilton and Milner~\cite{HM67} determines the largest size of a non-trivial intersecting family in this setting and characterizes the extremal families attaining it.
\begin{theorem}[Hilton--Milner Theorem~\cite{HM67}]\label{thm:HM}
For integers $k \geq 3$ and $n \geq 2k$, let $\calF \subseteq \binom{[n]}{k}$ be a non-trivial intersecting family.
Then,
\[ |\calF| \leq \binom{n-1}{k-1} - \binom{n-k-1}{k-1}+1.\]
Moreover, if $n>2k$ then equality holds if and only if there exist an element $i \in [n]$ and a $k$-subset $A$ of $[n]$ with $i \notin A$ such that
$\calF = \Big \{ F \in \binom{[n]}{k}~ \Big{|}~i \in F,~F \cap A \neq \emptyset \Big \} \cup \{A\}$,
or $k=3$ and there exists a $3$-subset $A$ of $[n]$ such that
$\calF = \Big \{F \in \binom{[n]}{3}~\Big{|}~|F \cap A| \geq 2 \Big \}$.
\end{theorem}

\subsection{Complexity Classes}

The complexity class $\TFNP$ consists of the total search problems in $\NP$, i.e., the search problems in which every input has a solution, where a solution can be verified in polynomial time. The complexity class $\PPA$ (Polynomial Parity Argument~\cite{Papa94}) consists of the problems in $\TFNP$ that can be reduced in polynomial time to a problem called $\LEAF$. The definition of the $\LEAF$ problem is not needed in this paper, but we mention it briefly below for completeness.

The $\LEAF$ problem asks, given a graph with maximum degree $2$ and a leaf (i.e., a vertex of degree $1$), to find another leaf in the graph. The input graph, though, is not given explicitly. Instead, the vertex set of the graph is defined to be $\{0,1\}^n$ for some integer $n$, and the graph is succinctly represented by a Boolean circuit that for a vertex of the graph computes its (at most two) neighbors. Note that the size of the graph might be exponential in the size of its description.

\subsection{Computational Problems}\label{sec:comput_prob}

We gather here several computational problems that will be studied and used throughout the paper.
We start with a computational search problem associated with Schrijver graphs. This problem is studied in Section~\ref{sec:Sch}.

\begin{definition}[Generalized Schrijver Problem]\label{def:SchrijverProb}
For $m = m(n,k)$, the $\SchrijverP(n,k,m)$ problem is defined as follows.
The input consists of two integers $n$ and $k$ with $n \geq 2k$ and a coloring $c: \binom{[n]}{k}_{\mathrm{stab}} \rightarrow [m]$ of the vertices of the graph $S(n,k)$ with $m=m(n,k)$ colors, and the goal is to find a monochromatic edge, i.e., two vertices $A,B \in \binom{[n]}{k}_{\mathrm{stab}}$ such that $A \cap B = \emptyset$ and $c(A)=c(B)$.
In the black-box input model, the coloring $c$ is given as an oracle access that given a vertex $A$ outputs its color $c(A)$. In the white-box input model, the coloring $c$ is given by a Boolean circuit that for a vertex $A$ computes its color $c(A)$.
For $m = n-2k+1$, the problem $\SchrijverP(n,k,m)$ is denoted by $\SchrijverP$.
\end{definition}

The $\KneserP$ problem is defined similarly to the $\SchrijverP$ problem. Here, the input coloring $c: \binom{[n]}{k} \rightarrow [n-2k+1]$ is defined on the entire vertex set of $K(n,k)$. By Theorem~\ref{thm:chi_KorS(n,k)}, every input of the $\SchrijverP$ and $\KneserP$ problems is guaranteed to have a solution. Moreover, whenever $m = m(n,k)  \leq n-2k+1$, every input of the $\SchrijverP(n,k,m)$ problem has a solution as well.

We remark that algorithms for the $\SchrijverP(n,k,m)$ problem are considered in this paper with respect to the black-box input model.
The running time of such an algorithm is referred to as polynomial if it is polynomial in $n$.
Observe that a polynomial-time algorithm for the $\SchrijverP(n,k,m)$ problem in the black-box input model yields an algorithm for the analogue problem in the white-box input model, whose running time is polynomial as well (in the input size).
For computational complexity results, like reductions and $\PPA$-completeness, we adopt the more suitable white-box input model.
For example, the $\SchrijverP$ problem in the white-box input model was shown in~\cite{Haviv22-FISC} to be $\PPA$-complete.

Another search problem studied in~\cite{Haviv22-FISC} is the following.

\begin{definition}[Fair Independent Set in Cycle Problem]\label{def:FISC}
In the $\FISC$ problem, the input consists of integers $n$ and $m$ along with a partition $V_1, \ldots ,V_m$ of $[n]$ into $m$ sets. The goal is to find a stable subset $S$ of $[n]$ satisfying $|S \cap V_i| \geq \frac{1}{2} \cdot |V_i|-1$ for all $i \in [m]$.
\end{definition}
\noindent
The existence of a solution for every input of the $\FISC$ problem was proved in~\cite{AABCKLZ17}.
It was shown in~\cite{Haviv22-FISC} that the $\FISC$ problem is $\PPA$-complete, even restricted to instances in which each part $V_i$ of the given partition has an odd size larger than $2$.

We next define the $\Uncovered$ problem, studied in Section~\ref{sec:Constrained}.
\begin{definition}[Unfair Independent Set in Cycle Problem]\label{def:Uncovered}
The input of the $\Uncovered$ problem consists of two integers $n$ and $k$ with $n \geq 2k$ and $\ell$ subsets $V_1, \ldots, V_{\ell}$ of $[n]$, where $\ell \leq n-2k+1$ and $|V_i| \geq 2$ for all $i \in [\ell]$. The goal is to find a stable $k$-subset $S$ of $[n]$ that satisfies the constraints $|S \cap V_i| \leq |V_i|/2$ for $i \in [\ell]$.
\end{definition}
\noindent
Note that Definition~\ref{def:Uncovered} requires the sets $V_1, \ldots, V_{\ell}$ of an instance of the $\Uncovered$ problem to satisfy $|V_i| \geq 2$ for all $i \in [\ell]$.
This requirement is justified by the observation that if $|V_i|=1$ for some $i \in [\ell]$, then any solution for this instance does not include the single element of $V_i$. Hence, by removing this element from the given sets and from the ground set, such an instance can be reduced to an instance with ground set of size smaller by one.
By repeatedly applying this reduction, one can get a `core' instance that fits Definition~\ref{def:Uncovered}.

We observe that the $\Uncovered$ problem is total. The argument relies on the chromatic number of the graph $S(n,k)$.
\begin{lemma}\label{lemma:Unco_total}
Every instance of the $\Uncovered$ problem has a solution.
\end{lemma}

\begin{proof}
Consider an instance of the $\Uncovered$ problem, i.e., integers $n$ and $k$ with $n \geq 2k$ and $\ell$ subsets $V_1, \ldots, V_{\ell}$ of $[n]$, where $\ell \leq n-2k+1$ and $|V_i| \geq 2$ for all $i \in [\ell]$.
For every $i \in [\ell]$, let
\[\calF_i = \bigg \{ S \in \binom{[n]}{k}_{\mathrm{stab}} ~\Big{|}~ |S \cap V_i| > |V_i|/2 \bigg \},\]
and notice that every two sets of $\calF_i$ have a common element of $V_i$, hence $\calF_i$ is an intersecting family. However, by Theorem~\ref{thm:chi_KorS(n,k)}, the chromatic number of $S(n,k)$ is $n-2k+2$, hence the family of stable $k$-subsets of $[n]$ cannot be covered by fewer than $n-2k+2$ intersecting families. By $\ell \leq n-2k+1$, this implies that there exists a set $S \in \binom{[n]}{k}_{\mathrm{stab}}$ that does not belong to any of the families $\calF_i$, hence it satisfies $|S \cap V_i| \leq |V_i|/2$ for all $i \in [\ell]$. This implies that $S$ is a valid solution for the given instance, and we are done.
\end{proof}

We end this section with the definition of the $\CT$ problem.
\begin{definition}[Cycle plus Triangles Problem]\label{def:CT}
In the $\CT$ problem, the input consists of an integer $k$ and a graph $G$ on $3k$ vertices, whose edge set is the disjoint union of a Hamilton cycle and $k$ pairwise vertex-disjoint triangles. The goal is to find an independent set in $G$ of size $k$.
\end{definition}
\noindent
The existence of a solution for every input of the $\CT$ problem follows from a result of~\cite{FleischnerS92} (see also~\cite{Sachs93,AABCKLZ17}).

\section{The Generalized Schrijver Problem}\label{sec:Sch}

In this section, we prove our results on the $\SchrijverP(n,k,m)$ problem (see Definition~\ref{def:SchrijverProb}).
We start with Theorem~\ref{thm:Sch_m_alpha}.

\begin{proof}[ of Theorem~\ref{thm:Sch_m_alpha}]
Fix some integer $d \geq 2$.
For integers $n$ and $k$ with $n \geq 2k$, put $t = \lfloor \frac{n}{2k+d-2} \rfloor$ and $m=d \cdot t-1$, and consider an instance of the $\SchrijverP(n,k,m)$ problem, i.e., a coloring $c: \binom{n}{k}_{\mathrm{stab}} \rightarrow [m]$ of the vertices of $S(n,k)$. The definition of $t$ allows us to consider $t$ pairwise disjoint subsets $J_1, \ldots, J_t$ of $[n]$, where each of the subsets includes $2k+d-2$ consecutive elements.
For each $i \in [t]$, let $\calS_i$ denote the family of all stable $k$-subsets of $J_i$ with respect to the natural cyclic order of $J_i$ (where the largest element precedes the smallest one), and notice that $\calS_i \subseteq \binom{n}{k}_{\mathrm{stab}}$.
Consider the algorithm that given an oracle access to a coloring $c$ as above, queries the oracle for the colors of all the sets of $\calS_1 \cup \cdots \cup \calS_t$, and returns a pair of disjoint sets from this collection that are assigned the same color by $c$.

For correctness, we show that the collection of sets $\calS_1 \cup \cdots \cup \calS_t$ necessarily includes two vertices that form a monochromatic edge.
Indeed, since the number of colors used by the coloring $c$ does not exceed $d \cdot t -1$, it follows that either there exist distinct $i,j \in [t]$ for which a vertex of $\calS_i$ and a vertex of $\calS_j$ have the same color, or there exists an $i \in [t]$ for which the vertices of $\calS_i$ are colored using fewer than $d$ colors.
For the former case, notice that for distinct $i$ and $j$, every vertex of $\calS_i$ is disjoint from every vertex of $\calS_j$, hence the collection includes two vertices that form a monochromatic edge. For the latter case, let $i \in [t]$ be an index for which the vertices of $\calS_i$ are colored using fewer than $d$ colors.
Observe that the subgraph of $S(n,k)$ induced by $\calS_i$ is isomorphic to the graph $S(2k+d-2,k)$, hence by Theorem~\ref{thm:chi_KorS(n,k)}, its chromatic number is $(2k+d-2)-2k+2 = d$.
Since the vertices of $\calS_i$ are colored using fewer than $d$ colors, it follows that they include two vertices that form a monochromatic edge, and we are done.

We finally analyze the running time of the algorithm. By Lemma~\ref{lemma:|S(n,k)|}, the number of vertices in the graph $S(n,k)$ is
\[\frac{n}{k} \cdot \binom{n-k-1}{k-1} = \frac{n}{k} \cdot \binom{n-k-1}{n-2k} \leq n \cdot (n-k-1)^{n-2k} \leq n^{n-2k+1}.\]
Since the subgraph of $S(n,k)$ induced by each $\calS_i$ is isomorphic to $S(2k+d-2,k)$, it follows that the total number of queries that the algorithm makes does not exceed $t \cdot (2k+d-2)^{d-1} \leq n^{O(d)}$. This implies that in running time $n^{O(d)}$, it is possible to enumerate all the sets of $\calS_1 \cup \cdots \cup \calS_t$, to query the oracle for their colors, and to find the desired monochromatic edge. This completes the proof.
\end{proof}

We consider now the $\SchrijverP(n,k,m)$ problem with $m = \lfloor n/2 \rfloor-2k+1$.
We first prove Theorem~\ref{thm:Schr<Kneser} that says that the problem is efficiently reducible to the $\KneserP$ problem (whose definition is given in Section~\ref{sec:comput_prob}).

\begin{proof}[ of Theorem~\ref{thm:Schr<Kneser}]
Put $m = \lfloor n/2 \rfloor-2k+1$, and let $c: \binom{[n]}{k}_{\mathrm{stab}} \rightarrow [m]$ be an instance of the $\SchrijverP(n,k,m)$ problem.
Consider the reduction that maps such a coloring $c$ to a coloring $c': \binom{[n]}{k} \rightarrow [n-2k+1]$ of the vertices of $K(n,k)$ defined as follows.
For every set $A \in \binom{[n]}{k}$, if $A$ is unstable then it includes an odd element, so denote its smallest odd element by $2i-1$, and define $c'(A)=i$. Notice that this $i$ satisfies $1 \leq i \leq \lceil n/2 \rceil$.
Otherwise, $A$ is a stable $k$-subset of $[n]$, and we define $c'(A) = c(A)+\lceil n/2 \rceil$.
Notice that $m+\lceil n/2 \rceil = n-2k+1$, hence the colors used by $c'$ are all in $[n-2k+1]$, as needed for an instance of the $\KneserP$ problem.
Notice further that given a Boolean circuit that computes the coloring $c$, it is possible to efficiently produce a Boolean circuit that computes the coloring $c'$.

For correctness, we simply show that any solution for the produced instance of the $\KneserP$ problem is also a solution for the given instance of the $\SchrijverP(n,k,m)$ problem.
To see this, consider a solution for the former, i.e., two disjoint $k$-subsets $A$ and $B$ of $[n]$ with $c'(A) = c'(B)$.
By the definition of $c'$, the color assigned by $c'$ to $A$ and $B$ cannot be some $i \leq \lceil n/2 \rceil$ because this would imply that the element $2i-1$ belongs to both $A$ and $B$, which are disjoint. It thus follows that $A$ and $B$ are stable $k$-subsets of $[n]$ satisfying $c'(A) = c(A) + \lceil n/2 \rceil$ and $c'(B) = c(B) + \lceil n/2 \rceil$. By $c'(A)=c'(B)$, it follows that $c(A)=c(B)$, hence $A$ and $B$ form a monochromatic edge in $S(n,k)$ and thus a solution for the given instance of the $\SchrijverP(n,k,m)$ problem. This completes the proof.
\end{proof}

The reduction presented in the proof of Theorem~\ref{thm:Schr<Kneser} extends a given coloring of $S(n,k)$ to a coloring of the entire graph $K(n,k)$. To do so, it uses a proper coloring with $\lceil n/2 \rceil$ colors of the subgraph $U(n,k)$ of $K(n,k)$ induced by the {\em unstable} $k$-subsets of $[n]$. However, in order to obtain a coloring of $K(n,k)$ with $n-2k+1$ colors, as required for instances of the $\KneserP$ problem, one has to reduce from the $\SchrijverP(n,k,m)$ problem with $m = \lfloor n/2 \rfloor-2k+1$. This suggests the question of whether $U(n,k)$ can be properly colored using fewer colors. Motivated by this question, we study some properties of this graph in Section~\ref{sec:unstable}, where we essentially answer this question in the negative (see Corollary~\ref{cor:chiU(n,k)} and the discussion that follows it).

We next show that the $\SchrijverP(n,k,m)$ problem with $m = \lfloor n/2 \rfloor-2k+1$ is not harder than the restriction of the standard $\SchrijverP$ problem to colorings of $S(n,k)$ with $n=4k$.
This confirms Theorem~\ref{thm:Sch_4k}.

\begin{proof}[ of Theorem~\ref{thm:Sch_4k}]
Suppose that there exists a polynomial-time algorithm, called $\mathsf{Algo}$, for the restriction of the $\SchrijverP$ problem to colorings of $S(n,k)$ with $n=4k$.
Such an algorithm is able to efficiently find a monochromatic edge in the graph $S(4k,k)$ given an access to a coloring of its vertices with fewer than $\chi(S(4k,k))$ colors.
By Theorem~\ref{thm:chi_KorS(n,k)}, it holds that $\chi(S(4k,k)) = 2k+2$.
Suppose without loss of generality that the algorithm $\mathsf{Algo}$ queries the oracle for the colors of the two vertices of the monochromatic edge that it returns.

For integers $n$ and $k$ with $n \geq 4k$, put $m = \lfloor n/2 \rfloor -2k+1$, and let $c: \binom{n}{k}_{\mathrm{stab}} \rightarrow [m]$ be an instance of the $\SchrijverP(n,k,m)$ problem, i.e., a coloring of the vertices of $S(n,k)$ with $m$ colors.
We present an algorithm that finds a monochromatic edge in $S(n,k)$. It may be assumed that $n > 8k$. Indeed, otherwise it holds that $m \leq 2k+1 < \chi(S(4k,k))$, hence a monochromatic edge can be found by running the given algorithm $\mathsf{Algo}$ on the restriction of the coloring $c$ to the subgraph of $S(n,k)$ induced by the stable $k$-subsets of $[4k]$. Since this graph is isomorphic to $S(4k,k)$, $\mathsf{Algo}$ is guaranteed to find a monochromatic edge in this subgraph, which also forms a monochromatic edge in the entire graph $S(n,k)$.

Now, put $t = \lfloor \frac{n}{4k} \rfloor$, and let $J_1, \ldots, J_t$ be $t$ pairwise disjoint subsets of $[n]$, where each of the subsets includes $4k$ consecutive elements.
For each $i \in [t]$, let $\calS_i$ denote the family of all stable $k$-subsets of $J_i$ with respect to the natural cyclic order of $J_i$ (where the largest element precedes the smallest one). Observe that the subgraph of $S(n,k)$ induced by the vertices of each $\calS_i$ is isomorphic to $S(4k,k)$.
Observe further that
\begin{eqnarray}\label{eq:m<t(2k+2)}
t \cdot (2k+2) > \Big (\frac{n}{4k}-1 \Big ) \cdot (2k+2) = \frac{n}{2}+\frac{n}{2k}-2k-2 > \Big \lfloor \frac{n}{2} \Big \rfloor -2k+1 = m,
\end{eqnarray}
where the last inequality holds because $n > 8k$.

Consider the algorithm that given an oracle access to a coloring $c$ as above, for each $i \in [t]$, simulates the algorithm $\mathsf{Algo}$ on the restriction of the coloring $c$ to the subgraph of $S(n,k)$ induced by the vertices of $\calS_i$.
If all the vertices queried throughout the $i$th simulation have at most $2k+1$ distinct colors, then the algorithm returns the monochromatic edge returned by $\mathsf{Algo}$. Otherwise, for each $i \in [t]$, the algorithm uses the queries made in the $i$th simulation of $\mathsf{Algo}$ to produce a set $\calF_i \subseteq \calS_i$ of $2k+2$ vertices with distinct colors.
Then, the algorithm finds a monochromatic edge that involves two vertices of $\calF_1 \cup \cdots \cup \calF_t$ and returns it.
This completes the description of the algorithm.
Since the running time of $\mathsf{Algo}$ is polynomial, the described algorithm can be implemented in polynomial time.

Let us prove the correctness of the algorithm.
Suppose first that for some $i \in [t]$, all the vertices queried throughout the $i$th simulation of $\mathsf{Algo}$ have at most $2k+1$ distinct colors (including the two vertices of the returned edge).
In this case, the answers of the oracle in the $i$th simulation is consistent with some coloring with at most $2k+1$ colors of the subgraph of $S(n,k)$ induced by $\calS_i$. Since this graph is isomorphic to $S(4k,k)$, whose chromatic number is $2k+2$, $\mathsf{Algo}$ is guaranteed to find in the graph a monochromatic edge, which is also a monochromatic edge in $S(n,k)$, and thus a valid output of the algorithm.
Otherwise, for each $i \in [t]$, the attempt to simulate $\mathsf{Algo}$ on the subgraph of $S(n,k)$ induced by $\calS_i$ provides a set $\calF_i$ of $2k+2$ vertices of $\calS_i$ with distinct colors.
By~\eqref{eq:m<t(2k+2)}, the total number $m$ of colors used by the coloring $c$ is smaller than $t \cdot (2k+2)$. This implies that there exist distinct indices $i,j \in [t]$ for which a vertex of $\calF_i$ and a vertex of $\calF_j$ have the same color. Since the vertices of $\calS_i$ are disjoint from those of $\calS_j$, these two vertices form a monochromatic edge in $S(n,k)$ and form a valid output of the algorithm.
\end{proof}

We end this section with the observation that there exists an efficient randomized algorithm for the $\SchrijverP(n,k,m)$ problem with $m = \lfloor n/2 \rfloor-2k+1$ on instances with $n = \Omega( k^4)$.
This follows from the paper~\cite{Haviv22b}, which yields that for such $n$ and $k$, the $\SchrijverP(n,k,m)$ problem is essentially reducible to the $\SchrijverP(n-1,k,m-1)$ problem in randomized polynomial time (with exponentially small failure probability).
By applying this reduction $m-1$ times, it follows that the $\SchrijverP(n,k,m)$ problem with $m = \lfloor n/2 \rfloor-2k+1$ where $n = \Omega(k^4)$, is efficiently reducible to the $\SchrijverP(\lceil n/2 \rceil +2k,k,1)$ problem, which can obviously be solved efficiently.

\section{The Unfair Independent Set in Cycle Problem}\label{sec:Constrained}

In this section, we study the $\Uncovered$ problem (see Definition~\ref{def:Uncovered}).

\subsection{Hardness}

We prove now Theorem~\ref{thm:Unc_PPA}, which asserts that $\Uncovered$ is $\PPA$-complete.

\begin{proof}[ of Theorem~\ref{thm:Unc_PPA}]
We first show that the $\Uncovered$ problem belongs to $\PPA$.
To do so, we show a polynomial-time reduction to the $\SchrijverP$ problem in the white-box input model, which lies in $\PPA$~\cite{Haviv22-FISC} (see Definition~\ref{def:SchrijverProb}).

Consider an instance of the $\Uncovered$ problem, i.e., integers $n$ and $k$ with $n \geq 2k$ and $\ell$ subsets $V_1, \ldots, V_{\ell}$  of $[n]$, where $\ell \leq n-2k+1$ and $|V_i| \geq 2$ for all $i \in [\ell]$. For such an instance, the reduction produces a Boolean circuit that given a stable $k$-subset $A$ of $[n]$, outputs the smallest index $i \in [\ell]$ such that $|A \cap V_i| > |V_i|/2$ if such an $i$ exists, and outputs $\ell$ otherwise. Note that this circuit represents a coloring $c: \binom{[n]}{k}_{\mathrm{stab}} \rightarrow [\ell]$ of the vertices of the graph $S(n,k)$ with $\ell \leq n-2k+1$ colors, hence it is an appropriate instance of the $\SchrijverP$ problem. Clearly, the Boolean circuit that computes $c$ can be constructed in polynomial time.

For correctness, we show that a solution for the constructed $\SchrijverP$ instance can be used to efficiently find a solution for the given $\Uncovered$ instance.
Consider a monochromatic edge of $S(n,k)$, i.e., two disjoint sets $A,B \in \binom{[n]}{k}_{\mathrm{stab}}$ with $c(A)=c(B)$. Since $A$ and $B$ are disjoint, it is impossible that $|A \cap V_i| > |V_i|/2$ and $|B \cap V_i| > |V_i|/2$ for some $i \in [\ell]$. By the definition of the coloring $c$, it follows that $c(A)=c(B)=\ell$, hence $|A \cap V_i| \leq |V_i|/2$ and $|B \cap V_i| \leq |V_i|/2$ for all $i \in [\ell-1]$. Moreover, at least one of $A$ and $B$ intersects $V_\ell$ at no more than $|V_\ell|/2$ elements, and thus forms a valid solution for the given $\Uncovered$ instance. Since it is possible to check in polynomial time which of the sets $A$ and $B$ satisfies this requirement, the proof of the membership of $\Uncovered$ in $\PPA$ is completed.

We next prove that the $\Uncovered$ problem is $\PPA$-hard. To do so, we reduce from the $\FISC$ problem (see Definition~\ref{def:FISC}). We use here the fact, proved in~\cite{Haviv22-FISC}, that this problem is $\PPA$-hard even when it is restricted to the instances in which the parts of the given partition have odd sizes larger than $2$.
Consider such an instance of the $\FISC$ problem, i.e., integers $n$ and $m$ along with a partition $V_1, \ldots ,V_m$ of $[n]$ such that $|V_i|$ is odd and satisfies $|V_i| \geq 3$ for all $i \in [m]$. Notice that $n$ and $m$ have the same parity, and define $k = \frac{n-m}{2}$. Our reduction simply returns the integers $n$ and $k$, which clearly satisfy $n \geq 2k$, and the sets $V_1, \ldots, V_m$. Note that $|V_i| \geq 2$ for all $i \in [m]$ and that the number $m$ of sets is $n-2k$. Since the latter does not exceed $n-2k+1$, this is a valid instance of the $\Uncovered$ problem.

For correctness, we show that a solution for the constructed $\Uncovered$ instance is also a solution for the given $\FISC$ instance.
Let $S$ be a solution for the $\Uncovered$ instance, i.e., a stable $k$-subset of $[n]$ such that for all $i \in [m]$ it holds that $|S \cap V_i| \leq |V_i|/2$. Since the sizes of the sets $V_1, \ldots, V_m$ are odd, it follows that $|S \cap V_i| \leq \frac{|V_i|-1}{2}$ for all $i \in [m]$. Since the sets $V_1, \ldots, V_m$ form a partition of $[n]$, it further follows that
\begin{eqnarray}\label{eq:|S|_reduction}
|S| = \sum_{i \in [m]}{|S \cap V_i|} \leq \sum_{i \in [m]}{\tfrac{|V_i|-1}{2}} = \frac{n-m}{2}=k.
\end{eqnarray}
However, by $|S|=k$, we derive from~\eqref{eq:|S|_reduction} that $|S \cap V_i| = \frac{|V_i|-1}{2}$ for all $i \in [m]$.
This implies that $S$ is a stable subset of $[n]$ satisfying $|S \cap V_i| \geq |V_i|/2-1$ for all $i \in [m]$, hence it forms a valid solution for the given $\FISC$ instance.
This completes the proof.
\end{proof}

Given the $\PPA$-hardness of the $\Uncovered$ problem, it is interesting to identify the range of the parameters $n$ and $k$ for which the hardness holds.
One can verify, using properties of the hard instances constructed in~\cite{Haviv22-FISC}, that the hardness given in Theorem~\ref{thm:Unc_PPA} holds for instances with $n = (2+o(1)) \cdot k$, where the $o(1)$ term tends to $0$ as $n$ and $k$ tend to infinity.
The following simple result shows that for $n=3k$ the problem is at least as hard as the $\CT$ problem, whose tractability is an open question (see Definition~\ref{def:CT}).

\begin{proposition}\label{prop:CT<Un}
The $\CT$ problem is polynomial-time reducible to the restriction of the $\Uncovered$ problem to instances that consist of $k$ sets of size $3$ that form a partition of $[n]$ where $n=3k$.
\end{proposition}

\begin{proof}
Consider an instance of the $\CT$ problem, i.e., an integer $k$ and a graph $G$ on $3k$ vertices, whose edge set is the disjoint union of a Hamilton cycle and $k$ pairwise vertex-disjoint triangles. It was shown in~\cite{Berczi017} that given such a graph it is possible to find in polynomial time such Hamilton cycle and triangles.
Let $V_1, \ldots, V_k$ denote the triplets of vertices of the $k$ triangles, and assume without loss of generality that the vertices along the Hamilton cycle are labeled by the elements of $[3k]$ according to their natural cyclic order.
Consider the polynomial-time reduction that given such an instance, returns the integers $n=3k$ and $k$ along with the sets $V_1, \ldots, V_k$. Observe that the number $k$ of sets does not exceed $n-2k+1=k+1$, hence the reduction returns an appropriate instance of the $\Uncovered$ problem.
Note that the reduction can be implemented in polynomial time.

For correctness, consider a solution for the produced instance, i.e., a stable $k$-subset $S$ of $[n]$ satisfying $|S \cap V_i| \leq 3/2$, and thus $|S \cap V_i| \leq 1$, for all $i \in [k]$. Since the sets $V_1, \ldots, V_k$ form a partition of $[n]$, using $|S|=k$, it follows that $|S \cap V_i| = 1$ for all $i \in [k]$. Therefore, the stable set $S$ is an independent set of size $k$ in $G$, and thus forms a solution for the given $\CT$ instance as well.
\end{proof}

\subsection{Algorithms}

We next prove Theorem~\ref{thm:alg_c}, which states that the $\Uncovered$ problem can be solved efficiently on instances with $n \geq c \cdot k$ for some absolute constant $c$.

\begin{proof}[ of Theorem~\ref{thm:alg_c}]
We start by presenting a randomized algorithm, based on a probabilistic argument with alterations, and then derandomize it using the method of conditional expectations.

Consider an instance of the $\Uncovered$ problem, i.e., integers $n$ and $k$ with $n \geq 2k$ and $\ell$ subsets $V_1, \ldots, V_\ell$ of $[n]$, where $\ell \leq n-2k+1$ and $|V_i| \geq 2$ for all $i \in [\ell]$. Put $r_i = |V_i| \geq 2$ for each $i \in [\ell]$. Suppose further that $n \geq c \cdot k$ for a sufficiently large constant $c$ to be determined later.
Let $p = 2k/n \leq 2/c$, and consider the following randomized algorithm.
\begin{enumerate}
  \item\label{itm:1} Pick a random subset $A$ of $[n]$ by including in $A$ every element of $[n]$ independently with probability $p$.
  \item\label{itm:2} Remove from $A$ every element $j \in [n]$ that satisfies $\{j,j+1\} \subseteq A$ (where for $j=n$, the element $j+1$ is considered as $1$). Let $A'$ denote the obtained set.
  \item\label{itm:3} For every $i \in [\ell]$ that satisfies $|A' \cap V_i|>r_i/2$, remove from $A'$ arbitrary $|A' \cap V_i| - \lfloor r_i/2 \rfloor$ elements of $V_i$. Let $A''$ denote the obtained set.
  \item\label{itm:4} If $|A''| \geq k$, then return an arbitrary $k$-subset of $A''$. Otherwise, return `failure'.
\end{enumerate}

We first claim that unless the algorithm returns `failure', it returns a valid output. Indeed, Item~\ref{itm:2} of the algorithm guarantees that the set $A'$ is stable. Further, Item~\ref{itm:3} guarantees that its subset $A''$ satisfies $|A'' \cap V_i| \leq \lfloor r_i/2 \rfloor$ for all $i \in [\ell]$. Therefore, in the case where $|A''| \geq k$, any $k$-subset of $A''$ returned in Item~\ref{itm:4} of the algorithm is a valid solution for the given $\Uncovered$ instance.

We next estimate the expected size of the set $A''$ produced by the algorithm.
The set $A$ chosen in Item~\ref{itm:1} of the algorithm includes every element of $[n]$ with probability $p$. Hence, its expected size satisfies $\Expec{}{|A|} = p \cdot n$. In Item~\ref{itm:2} of the algorithm, the probability of every element of $[n]$ to be removed from $A$ is equal to the probability that both the element and its successor modulo $n$ belong to $A$, which is $p^2$. By linearity of expectation, this implies that the expected size of the set $A'$ satisfies $\Expec{}{|A'|} = (p-p^2) \cdot n$.
It remains to estimate the expected number of elements removed from $A'$ in Item~\ref{itm:3} of the algorithm.
Observe that for each $i \in [\ell]$, the algorithm removes from $A'$ the smallest possible number of elements of $V_i$ ensuring that the obtained set $A''$ includes at most $\lfloor r_i/2 \rfloor$ of them.
Therefore, the number of removed elements of $V_i$ does not exceed the number of subsets of $V_i$ of size $\lfloor r_i/2 \rfloor + 1$ that are contained in $A$ (because it suffices to remove one element from each of them). It thus follows that the expected number of elements of $V_i$ that are removed from $A'$ in Item~\ref{itm:3} of the algorithm is at most
\[ \binom{r_i}{\lfloor r_i/2 \rfloor+1} \cdot p^{\lfloor r_i/2 \rfloor + 1} \leq 2^{r_i} \cdot p^{\lfloor r_i/2 \rfloor+1} \leq (4p)^{\lfloor r_i/2 \rfloor + 1} \leq (4p)^2,\]
where in the last inequality we use the assumption $r_i \geq 2$ and the fact that $p \leq 1/4$ (which holds for any sufficiently large choice of the constant $c$).
It therefore follows, using again the linearity of expectation, that the expected size of $A''$ satisfies
\[ \Expec{}{|A''|} \geq (p-p^2) \cdot n - \ell \cdot (4p)^2 \geq (p-17p^2) \cdot n \geq k, \]
where the second inequality holds by $\ell \leq n$, and the last inequality by the definition of $p=2k/n$, assuming again that $n \geq c \cdot k$ for a sufficiently large constant $c$ (say, $c=68$). This implies that there exists a random choice for the presented randomized algorithm for which it returns a valid solution.

We next apply the method of conditional expectations to derandomize the above algorithm.
Let us start with a few notations.
For a set $S \subseteq [n]$, define
\begin{eqnarray}\label{eq:f}
f(S) = |S| - |\{ j \in [n] \mid \{j,j+1\} \subseteq S\}| - \sum_{i \in [\ell]}{\Big | \Big \{B \subseteq S \cap V_i ~\Big{|}~ |B| = \lfloor r_i/2 \rfloor+1 \Big \} \Big |}.
\end{eqnarray}
In words, $f(S)$ is determined by subtracting from the size of $S$ the number of pairs of consecutive elements in $S$ (modulo $n$) as well as the number of subsets of $S \cap V_i$ of size $\lfloor r_i/2 \rfloor+1$ for each $i \in [\ell]$.
For a vector $x \in \{0,1,\ast\}^n$, let $S_x$ denote a random subset of $[n]$ such that for every $i \in [n]$, if $x_i=1$ then $i \in S_x$, if $x_i=0$ then $i \notin S_x$, and if $x_i = \ast$ then $i$ is chosen to be included in $S_x$ independently with probability $p=2k/n$. We refer to the vector $x$ as a {\em partial choice} of a subset of $[n]$.
We further define a potential function $\phi:\{0,1,\ast\}^n \rightarrow \R$ that maps every vector $x \in \{0,1,\ast\}^n$ to the expected value of $f(S)$ where $S$ is chosen according to the distribution $S_x$, that is, $\phi(x) = \Expec{}{f(S_x)}$.

We observe that given a partial choice $x \in \{0,1,\ast\}^n$, the value of $\phi(x)$ can be calculated efficiently, in time polynomial in $n$.
Indeed, to calculate the expected value of $f(S_x)$, it suffices, by linearity of expectation, to calculate the expected value of each of the three terms in~\eqref{eq:f} evaluated at the set $S_x$. It is easy to see that the expected value of the first term is
\[|\{j \in [n] \mid x_j=1\}|+ p \cdot |\{ j \in [n] \mid x_j = \ast\}|,\]
and that the expected value of the second term is
\[ \Big | \Big \{ j \in [n] \Big{|} x_j=x_{j+1}=1 \Big \} \Big | + p \cdot \Big | \Big \{ j \in [n] \Big{|} \{x_j,x_{j+1}\}=\{1,\ast\} \Big \} \Big | + p^2 \cdot \Big | \Big \{ j \in [n] \Big{|} x_j=x_{j+1}=\ast \Big \} \Big |.\]
As for the third term, by linearity of expectation, it suffices to determine the expected value of
\[ \Big | \Big \{B \subseteq S_x \cap V_i ~\Big{|}~ |B| = \lfloor r_i/2 \rfloor+1 \Big \} \Big |\]
for $i \in [\ell]$. Letting $s_i = |\{ j \in V_i \mid x_j=\ast\}|$ and $t_i = |\{ j \in V_i \mid x_j=1\}|$, one can check that the required expectation is precisely
\[\sum_{m=0}^{\lfloor r_i/2 \rfloor+1}{ \binom{s_i}{m} \cdot \binom{t_i}{\lfloor r_i/2 \rfloor+1 - m} \cdot p^m}.\]
Since all the terms can be calculated in time polynomial in $n$, so can $\phi(x)$.

We describe a deterministic algorithm that finds a set $S \subseteq [n]$ satisfying $f(S) \geq k$.
Given such a set, the algorithm is completed by applying Items~\ref{itm:2},~\ref{itm:3}, and~\ref{itm:4} of the algorithm presented above.
Indeed, by applying Items~\ref{itm:2} and~\ref{itm:3} we obtain a stable set $S''$ such that $|S'' \cap V_i| \leq r_i/2$ for all $i \in [\ell]$.
The fact that $f(S) \geq k$ guarantees that this set $S''$ satisfies $|S''| \geq k$, hence Item~\ref{itm:4} returns a valid solution.

To obtain the desired set $S \subseteq [n]$ with $f(S) \geq k$, our algorithm maintains a partial choice $x \in \{0,1,\ast\}^n$ satisfying $\phi(x) \geq k$.
We start with $x = (\ast, \ldots, \ast)$, for which the analysis of the randomized algorithm guarantees that $\phi(x) \geq k$, provided that $n \geq c \cdot k$ for a sufficiently large constant $c$. We then choose the entries of $x$, one by one, to be either $0$ or $1$.
In the $i$th iteration, in which $x_1, \ldots, x_{i-1} \in \{0,1\}$, the algorithm evaluates $\phi$ at the two partial choices $x_{i \leftarrow 0} = (x_1, \ldots, x_{i-1},0,\ast,\ldots,\ast)$ and $x_{i \leftarrow 1} = (x_1, \ldots, x_{i-1},1,\ast,\ldots,\ast)$, and continues to the next iteration with one of them which maximizes the value of $\phi$. By the law of total expectation, it holds that $\phi(x) = p \cdot \phi(x_{i \leftarrow 1}) + (1-p) \cdot \phi(x_{i \leftarrow 0})$, implying that the choice of the algorithm preserves the inequality $\phi(x) \geq k$.
At the end of the process, we get a vector $x \in \{0,1\}^n$ with $\phi(x) \geq k$, which fully determines the desired set $S$ with $f(S) \geq k$.
Since the evaluations of $\phi$ can be calculated in time polynomial in $n$, the algorithm can be implemented in polynomial time.
This completes the proof.
\end{proof}

Given the above result, it would be interesting to determine the smallest constant $c$ for which the $\Uncovered$ problem can be solved efficiently on instances with $n \geq c \cdot k$. Of particular interest is the restriction of the problem to instances with $n=3k$ and with pairwise disjoint sets of size $3$, because as follows from Proposition~\ref{prop:CT<Un}, an efficient algorithm for this restriction would imply an efficient algorithm for the $\CT$ problem.
Interestingly, it turns out that the restriction of the $\Uncovered$ problem to instances with $n=4k$ and with pairwise disjoint sets of size $4$ does admit an efficient algorithm. This is a consequence of the following result derived from an argument of Alon~\cite{Alon92strong} (see also~\cite{AlonFair22}). We include its quick proof for completeness.

\begin{proposition}\label{prop:|V_i|=4}
There exists a polynomial-time algorithm that given an integer $k$ and a partition of $[4k]$ into $k$ subsets $V_1, \ldots, V_k$ with $|V_i|=4$ for all $i \in [k]$, finds a partition of $[4k]$ into four stable $k$-subsets $S_1,S_2,S_3,S_4$ of $[4k]$ such that $|S_j \cap V_i| = 1$ for all $j \in [4]$ and $i \in [k]$.
\end{proposition}

\begin{proof}
Consider the algorithm that given an integer $k$ and a partition of $V = [4k]$ into $k$ subsets $V_1, \ldots, V_k$ with $|V_i|=4$ for all $i \in [k]$, acts as follows.

Let $M_1$ and $M_2$ be the two matchings of size $2k$ in the cycle on the vertex set $V$ with the natural order along the cycle.
Let $M_3$ be an arbitrary matching of size $2k$ on the vertex set $V$ that includes two edges with vertices in $V_i$ for each $i \in [k]$.
Consider the graph $G_1 = (V,M_1 \cup M_3)$. Since the edge set of $G_1$ is a union of two matchings, it has no odd cycles, hence it is $2$-colorable, and a $2$-coloring $c_1$ of $G_1$ can be found in polynomial time.
Notice that the edges of $M_3$ in $G_1$ guarantee that each color class of the coloring $c_1$ includes precisely two vertices from $V_i$ for each $i \in [k]$.

Now, let $M_4$ be the matching of size $2k$ on the vertex set $V$ that includes, for each $i \in [k]$, the two edges with vertices in $V_i$ whose endpoints have the same color according to $c_1$.
Consider the graph $G_2 = (V,M_2 \cup M_4)$. As before, $G_2$ is $2$-colorable, and a $2$-coloring $c_2$ of $G_2$ can be found in polynomial time.

Finally, the algorithm defines a $4$-coloring of the graph $(V,M_1 \cup M_2 \cup M_3 \cup M_4)$ by assigning to every vertex $v \in V$ the pair $(c_1(v),c_2(v))$.
Observe that the produced coloring is proper and that its color classes can be found efficiently.
By the edges of $M_1 \cup M_2$, each color class is a stable subset of $[4k]$, and by the edges of $M_3 \cup M_4$, each color class includes at most one vertex from each $V_i$.
This implies that the four color classes are stable $k$-subsets of $V$ that satisfy the assertion of the proposition, and we are done.
\end{proof}

\section{Unstable Sets}\label{sec:unstable}

In this section, we explore two subgraphs of the Kneser graph $K(n,k)$ induced by families of unstable $k$-subsets of $[n]$.
These subgraphs are defined as follows.
\begin{definition}\label{def:U(n,k)}
Let $n$ and $k$ be integers with $n \geq 2k$.
Let $\widetilde{U}(n,k)$ denote the subgraph of $K(n,k)$ induced by the family of all $k$-subsets of $[n]$ that include a pair of consecutive elements (where the elements $n$ and $1$ are {\em not} considered as consecutive for $n>2$).
Let $U(n,k)$ denote the subgraph of $K(n,k)$ induced by the family of all $k$-subsets of $[n]$ that include a pair of consecutive elements modulo $n$, i.e., the family of unstable $k$-subsets of $[n]$.
\end{definition}

\subsection{Chromatic Number}

We study now the chromatic numbers of the graphs $U(n,k)$ and $\widetilde{U}(n,k)$.
It is worth mentioning here that a result of Do{\lsoft}nikov~\cite{Dolnikov82} generalizes the lower bound of Lov{\'{a}}sz~\cite{LovaszKneser} on the chromatic number of $K(n,k)$ to general graphs, using a notion called colorability defect (see also~\cite[Chapter~3.4]{MatousekBook} and~\cite{Kriz92}). This generalization implies a tight lower bound of $n-2k+2$ on the chromatic number of $K(n,k)$ and a somewhat weaker lower bound of $n-4k+4$ on the chromatic number of $S(n,k)$ (see, e.g.,~\cite{MatousekZ04}). It turns out, though, that this generalized approach of~\cite{Dolnikov82} does not yield any meaningful bounds on the chromatic numbers of the graphs from Definition~\ref{def:U(n,k)}.

The following theorem determines the exact chromatic number of the graph $\widetilde{U}(n,k)$.

\begin{theorem}\label{thm:chi_tilde_U}
For all integers $n$ and $k$ with $n \geq 2k$,
\[\chi(\widetilde{U}(n,k)) = \min(n-2k+2,\lfloor n/2 \rfloor).\]
\end{theorem}

The proof of Theorem~\ref{thm:chi_tilde_U} relies on a topological argument.
It uses the following variant of the Borsuk--Ulam theorem (see, e.g.,~\cite[Ex.~2.1.6]{MatousekBook}).
Here, $\Sset^t$ stands for the $t$-dimensional unit sphere with respect to the Euclidean norm, that is, $\Sset^t = \{ x \in \R^{t+1} \mid \|x\|=1\}$.

\begin{theorem}\label{thm:BU}
If the $t$-dimensional sphere $\Sset^t$ is covered by $t+1$ sets $F_1, \ldots, F_{t+1}$, each $F_j$ open or closed, then there exist an index $j \in [t+1]$ and a point $x \in \Sset^t$ such that both $x$ and $-x$ belong to $F_j$.
\end{theorem}

Another crucial ingredient in the proof of Theorem~\ref{thm:chi_tilde_U} is the following lemma.
It is inspired by a lemma of Gale~\cite{Gale56} that was applied in~\cite{SchrijverKneser78} to determine the chromatic number of $S(n,k)$ (see also~\cite[Chapter~3.5]{MatousekBook}).
Here, a hyperplane $h$ in $\R^{t+1}$ that passes through the origin is defined as the set $\{ z \in \R^{t+1} \mid \langle x,z \rangle = 0\}$ for some $x \in \Sset^t$, and the two open hemispheres that $h$ determines are $\{ z \in \Sset^t \mid \langle x,z\rangle >0 \}$ and $\{ z \in \Sset^t \mid \langle x,z\rangle <0 \}$.

\begin{lemma}\label{lemma:y_i_tilde_U}
For integers $n$ and $k$ with $n \geq 2k$, let $t = \min(n-2k+2,\lfloor n/2 \rfloor)-1$.
Then, there exist $n$ points $y_1,\ldots,y_n \in \Sset^t$ such that for every hyperplane $h$ in $\R^{t+1}$ that passes through the origin, at least one of the two open hemispheres that $h$ determines contains the points of $\{y_i \mid i \in A\}$ for some vertex $A$ of $\widetilde{U}(n,k)$.
\end{lemma}

\begin{proof}
Let $n$, $k$, and $t$ be integers as in the statement of the lemma.
Let $\gamma: \R \rightarrow \R^{t+1}$ denote the function defined by
\[ \gamma(x) = (1,x,x^2,\ldots,x^t). \]
For every $i \in [n]$, consider the point $w_i = \gamma(i) \in \R^{t+1}$.
We prove that for every hyperplane $h$ in $\R^{t+1}$ that passes through the origin, at least one of the two open half-spaces that $h$ determines contains the points of $\{w_i \mid i \in A\}$ for some vertex $A$ of $\widetilde{U}(n,k)$. This will immediately imply that the points $y_1, \ldots, y_n$ defined by $y_i = w_i/\|w_i\|$ for $i \in [n]$, which all lie on $\Sset^t$, satisfy the assertion of the lemma.

Consider an arbitrary hyperplane $h$ in $\R^{t+1}$ that passes through the origin.
Every point $w_i$ either lies on $h$ or belongs to one of the two open half-spaces determined by $h$.
Let $W_{on}$ denote the set of indices $i \in [n]$ for which the point $w_i$ lies on $h$, and let $W_1$ and $W_2$ denote the sets of indices $i \in [n]$ of the points $w_i$ that belong to the two open half-spaces determined by $h$. Our goal is to show that at least one of the sets $W_1$ and $W_2$ contains a vertex of $\widetilde{U}(n,k)$. To do so, one has to show that at least one of them includes two consecutive elements and has size at least $k$.

The definition of the points $w_1,\ldots,w_n$ implies that the size of $W_{on}$ does not exceed the number of roots of some nonzero polynomial of degree at most $t$, hence
$|W_{on}| \leq t$.
In fact, we may and will assume that $|W_{on}|=t$, as otherwise it is possible to continuously move $h$ so that it will satisfy this property and no $w_i$ will cross from one side to the other (see, e.g.,~\cite[Lemma~3.5.1]{MatousekBook}).
Now, the points $w_i$ with $i \in W_{on}$ divide the image $\Image(\gamma)$ of the function $\gamma$ into $t+1$ open continuous parts that alternate between the two open half-spaces determined by $h$.
Observe that all the indices $i$ of the points $w_i$ that belong to every such continuous part are either in $W_1$ or in $W_2$.
Suppose without loss of generality that $W_1$ corresponds to the points $w_i$ that belong to $\lceil \tfrac{t+1}{2} \rceil$ of the parts of $\Image(\gamma)$, and thus $W_2$ corresponds to the points $w_i$ that belong to the other $\lfloor \tfrac{t+1}{2} \rfloor$ parts of $\Image(\gamma)$.

In order to prove that $W_1$ (or $W_2$) includes two consecutive elements, it suffices to show that its size exceeds the number of parts of $\Image(\gamma)$ associated with it.
We first observe that at least one of the sets $W_1$ and $W_2$ satisfies this requirement, that is,
\begin{eqnarray}\label{eq:|W1|,|W2|}
|W_1| > \Big \lceil \frac{t+1}{2} \Big \rceil ~~~\mbox{or}~~~ |W_2| > \Big \lfloor \frac{t+1}{2} \Big \rfloor.
\end{eqnarray}
To see this, assume for the sake of contradiction that both the inequalities in~\eqref{eq:|W1|,|W2|} do not hold.
It follows that
\[n-t = |W_1|+|W_2| \leq \Big \lceil \frac{t+1}{2} \Big \rceil + \Big \lfloor \frac{t+1}{2} \Big \rfloor = t+1,\]
which implies that $n \leq 2 t +1$.
This, however, contradicts the definition of $t$ which guarantees that $t \leq \lfloor n/2 \rfloor -1$.

Given that at least one of the inequalities in~\eqref{eq:|W1|,|W2|} holds, it is easy to see that in the case where $|W_1| \leq |W_2|$, the first inequality in~\eqref{eq:|W1|,|W2|} implies the second, hence the second inequality in~\eqref{eq:|W1|,|W2|} necessarily holds.
Further, in the case where $|W_1| > |W_2|$, since the right-hand side of the two inequalities in~\eqref{eq:|W1|,|W2|} differ by at most $1$, the first inequality in~\eqref{eq:|W1|,|W2|} necessarily holds.
It thus follows that for at least one of the sets $W_1$ and $W_2$, its size is equal to $\max(|W_1|,|W_2|)$ and is strictly larger than the number of parts of $\Image(\gamma)$ associated with it.
This implies that this set includes two consecutive elements.
By the definition of $t$, we have $t \leq n-2k+1$, which implies that the size of the set satisfies
\begin{eqnarray}\label{eq:max_W1_W2}
\max(|W_1|,|W_2|) \geq \Big \lceil \frac{n-t}{2} \Big \rceil \geq \Big \lceil \frac{2k-1}{2} \Big \rceil = k.
\end{eqnarray}
This completes the proof.
\end{proof}

We are ready to prove Theorem~\ref{thm:chi_tilde_U}.

\begin{proof}[ of Theorem~\ref{thm:chi_tilde_U}]
For the upper bound, apply first Theorem~\ref{thm:chi_KorS(n,k)} to obtain that \[\chi(\widetilde{U}(n,k)) \leq \chi(K(n,k)) = n-2k+2.\]
Next, since every vertex of $\widetilde{U}(n,k)$ includes two consecutive elements, it must include an even element. By assigning to every such vertex its minimal even element, we obtain a proper coloring of $\widetilde{U}(n,k)$ with $\lfloor n/2 \rfloor$ colors, hence $\chi(\widetilde{U}(n,k)) \leq \lfloor n/2 \rfloor$. This completes the proof of the upper bound.

The lower bound relies on the Borsuk--Ulam theorem (Theorem~\ref{thm:BU}). Let \[t = \min (n-2k+2, \lfloor n/2 \rfloor) - 1,\] and suppose for the sake of contradiction that there exists a proper coloring of $\widetilde{U}(n,k)$ with $t$ colors. Let $y_1,\ldots,y_n \in \Sset^t$ denote the points given by Lemma~\ref{lemma:y_i_tilde_U}. We define $t$ sets $F_1,\ldots,F_t \subseteq \Sset^t$ as follows. A point $x \in \Sset^t$ is included in $F_j$ with $j \in [t]$ if there exists a vertex $A$ of $\widetilde{U}(n,k)$ colored $j$ such that $\{ y_i \mid i \in A\} \subseteq H(x)$, where $H(x)$ is the open hemisphere centered at $x$. We further define $F_{t+1} = \Sset^t \setminus (F_1 \cup \cdots \cup F_t)$. Note that the sets $F_1,\ldots,F_{t+1}$ cover $\Sset^t$. Note further that $F_1,\ldots,F_t$ are open whereas $F_{t+1}$ is closed.

By Theorem~\ref{thm:BU}, there exist an index $j \in [t+1]$ and a point $x \in \Sset^t$ such that both $x$ and $-x$ belong to $F_j$.
If $j \in [t]$, then it follows from the definition of $F_j$ that there exist two vertices of $\widetilde{U}(n,k)$ with color $j$ that correspond to disjoint sets, contradicting the assumption that the given coloring is proper.
If $j = t+1$ then neither $H(x)$ nor $H(-x)$ contains $\{ y_i \mid i \in A\}$ for a vertex $A$ of $\widetilde{U}(n,k)$, contradicting Lemma~\ref{lemma:y_i_tilde_U}.
This completes the proof.
\end{proof}

We derive the following result on the chromatic number of $U(n,k)$.

\begin{corollary}\label{cor:chiU(n,k)}
For all integers $n$ and $k$ with $n \geq 2k$,
\[\min(n-2k+2,\lfloor n/2 \rfloor) \leq \chi(U(n,k)) \leq \min(n-2k+2,\lceil n/2 \rceil).\]
\end{corollary}

\begin{proof}
For the upper bound, apply first Theorem~\ref{thm:chi_KorS(n,k)} to obtain that
\[\chi(U(n,k)) \leq \chi(K(n,k)) = n-2k+2.\]
Next, since every vertex of $U(n,k)$ includes two consecutive elements modulo $n$, it must include an odd element. By assigning to every such vertex its minimal odd element, we obtain a proper coloring of $U(n,k)$ with $\lceil n/2 \rceil$ colors, hence $\chi(U(n,k)) \leq \lceil n/2 \rceil$.
This completes the proof of the upper bound.
The lower bound follows by combining Theorem~\ref{thm:chi_tilde_U} with the fact that $\widetilde{U}(n,k)$ is an induced subgraph of $U(n,k)$.
\end{proof}

We conclude this section with a discussion on the tightness of Corollary~\ref{cor:chiU(n,k)}.
Notice that the upper and lower bounds provided in Corollary~\ref{cor:chiU(n,k)} coincide whenever the integer $n$ is even or satisfies $n \leq 4k-4$.
For other values of $n$ and $k$ the two bounds differ by $1$.
Yet, it turns out that the proof technique of Theorem~\ref{thm:chi_tilde_U} can be used to show that the upper bound in Corollary~\ref{cor:chiU(n,k)} is tight for all integers $n$ that are congruent to $1$ modulo $4$. We provide the details in Appendix~\ref{app:1mod4}. This leaves us with a gap of $1$ between the upper and lower bounds in Corollary~\ref{cor:chiU(n,k)} only for those integers $n$ and $k$, where $n$ is congruent to $3$ modulo $4$ and satisfies $n \geq 4k-1$.

We further observe that for an odd integer $n$ and for every proper coloring of $U(n,k)$ that includes a trivial color class (all of whose members share a common element), the number of used colors is at least the upper bound in Corollary~\ref{cor:chiU(n,k)}. Indeed, the restriction of such a coloring to the vertices that do not include the common element of the trivial color class is a proper coloring of a graph isomorphic to $\widetilde{U}(n-1,k)$, so by Theorem~\ref{thm:chi_tilde_U} it uses at least $\min(n-2k+1,(n-1)/2)$ colors. Together with the additional color of the trivial color class, the total number of colors is at least $\min(n-2k+2, \lceil n/2 \rceil)$, as claimed.

\subsection{Independence Number}

We next determine the largest size of an independent set in the graph $U(n,k)$.
The proof uses the Hilton--Milner theorem (Theorem~\ref{thm:HM}).

\begin{theorem}\label{thm:alphaU(n,k)}
For all integers $k \geq 2$ and $n \geq 2k$, it holds that
\[ \alpha(U(n,k)) = \binom{n-1}{k-1} - \binom{n-k-1}{k-1}.\]
\end{theorem}

\begin{proof}
Fix some integers $k \geq 2$ and $n \geq 2k$.
We first observe that there exists an independent set in $U(n,k)$ with the required size. This follows by considering the family of all vertices of $U(n,k)$ that include some fixed element $i \in [n]$. This family is clearly intersecting and thus forms an independent set in $U(n,k)$. The number of $k$-subsets of $[n]$ that include $i$ is $\binom{n-1}{k-1}$, where by Lemma~\ref{lemma:Talbot}, $\binom{n-k-1}{k-1}$ of them are stable and therefore do not form vertices of $U(n,k)$. Therefore,
\[\alpha(U(n,k)) \geq \binom{n-1}{k-1} - \binom{n-k-1}{k-1}.\]

We now prove the upper bound.
We start with the simple case of $n=2k$.
Here, the edges of the Kneser graph $K(n,k)$ form a perfect matching in which every edge connects a set to its complement. Any independent set in $U(n,k)$ includes at most one vertex from every edge in this matching. However, the set of the odd elements of $[n]$ and the set of the even elements of $[n]$ are adjacent in $K(n,k)$ but are not vertices of $U(n,k)$.
This implies that for $n=2k$, it holds that
\[\alpha(U(n,k)) \leq \frac{1}{2} \cdot \binom{2k}{k}-1 = \binom{2k-1}{k-1}-1.\]
This coincides with the required bound.

Suppose now that $k >3$ and $n > 2k$.
Let $\calF$ be a family of $k$-subsets of $[n]$ that forms an independent set in $U(n,k)$.
Since $U(n,k)$ is an induced subgraph of $K(n,k)$, it follows that $\calF$ is an intersecting family.
If $\calF$ is trivial, then its size does not exceed the number of vertices in $U(n,k)$ that include some fixed element $i \in [n]$, hence $|\calF| \leq \binom{n-1}{k-1} - \binom{n-k-1}{k-1}$, and we are done.
So suppose that $\calF$ is  non-trivial, and assume for the sake of contradiction that its size satisfies $|\calF| \geq \binom{n-1}{k-1} - \binom{n-k-1}{k-1} +1$.
By the Hilton--Milner theorem, stated as Theorem~\ref{thm:HM}, it follows that $|\calF| = \binom{n-1}{k-1} - \binom{n-k-1}{k-1} +1$. Moreover, using $k >3$, it follows that there exist an element $i \in [n]$ and a $k$-subset $A$ of $[n]$ with $i \notin A$ such that $\calF = \{ F \in \binom{[n]}{k}~\mid~i \in F,~F \cap A \neq \emptyset\} \cup \{A\}$.

To obtain a contradiction, it suffices to show that such a family includes a set that does not form a vertex of $U(n,k)$.
So assume without loss of generality that $i=1$, and consider the following two $k$-subsets of $[n]$: $\{1,3,5,\ldots,2k-1\}$ and $\{1,4,6,\ldots,2k\}$. By $n > 2k$, neither of them is a vertex of $U(n,k)$. If there exists an element $j \in A$ satisfying $3 \leq j \leq 2k$, then at least one of them belongs to $\calF$, hence $\calF$ includes a set that does not form a vertex of $U(n,k)$, as required.
Otherwise, using $k \geq 3$ and $1 \notin A$, there exists an element $j \in A$ satisfying $2k < j < n$. It thus follows that the set $\{1,3,5,\ldots,2k-3,j\}$ belongs to $\calF$ but does not form a vertex of $U(n,k)$, as required.

We turn to the case in which $k=3$ and $n > 6$.
The proof for this case follows the argument presented above for $k>3$. The only difference is that while assuming in contradiction that the independent set $\calF$ satisfies $|\calF| = \binom{n-1}{k-1} - \binom{n-k-1}{k-1} +1$, the Hilton--Milner theorem provides another possible structure for $\calF$, namely, $\calF = \{F \in \binom{[n]}{3}~\mid~|F \cap A| \geq 2\}$ for some $3$-subset $A$ of $[n]$. It thus remains to show that such a family must include a set that does not form a vertex of $U(n,3)$.
Note that $A \in \calF$. Therefore, if $A$ is not a vertex of $U(n,3)$ then we are done.
Otherwise, it can be assumed without loss of generality that $A = \{1,2,j\}$ for some $3 \leq j < n$.
For $j=3$, the set $\{1,3,5\}$ belongs to $\calF$ but does not form a vertex of $U(n,3)$.
For $j=4$, the set $\{1,4,6\}$ belongs to $\calF$ but does not form a vertex of $U(n,3)$.
And finally, for $j \geq 5$, the set $\{1,3,j\}$ belongs to $\calF$ but does not form a vertex of $U(n,3)$, so we are done.

We are left with the case of $k=2$.
Here, the graph $U(n,k)$ is isomorphic to the complement of a cycle on $n$ vertices. It is not difficult to check that for $n \geq 4$, it holds that $\alpha(U(n,2)) = 2$, which coincides with the stated result. This completes the proof.
\end{proof}

\bibliographystyle{abbrv}
\bibliography{few-colors}

\appendix

\section{On the Tightness of Corollary~\ref{cor:chiU(n,k)}}\label{app:1mod4}

We prove the following result.

\begin{theorem}\label{thm:U(n,k)_n=1mod4}
For all integers $n$ and $k$ such that $n$ is congruent to $1$ modulo $4$ and $n \geq 2k$,
\[ \chi(U(n,k)) = \min (n-2k+2, \lceil n/2 \rceil).\]
\end{theorem}
\noindent
Theorem~\ref{thm:U(n,k)_n=1mod4} shows that the upper bound in Corollary~\ref{cor:chiU(n,k)} is tight whenever $n$ is congruent to $1$ modulo $4$.
Its proof relies on the following lemma whose proof resembles that of Lemma~\ref{lemma:y_i_tilde_U}.

\begin{lemma}\label{lemma:y_i_U_odd}
For integers $n$ and $k$ for which it holds that $n$ is congruent to $1$ modulo $4$ and $n \geq 2k$, let $t = \min(n-2k+2,\lceil n/2 \rceil)-1$.
Then, there exist $n$ points $y_1,\ldots,y_n \in \Sset^t$ such that for every hyperplane $h$ in $\R^{t+1}$ that passes through the origin, at least one of the two open hemispheres that $h$ determines contains the points of $\{y_i \mid i \in A\}$ for some vertex $A$ of $U(n,k)$.
\end{lemma}

\begin{proof}
Let $n$, $k$, and $t$ be integers as in the statement of the lemma.
Observe that the assumption that $n$ is congruent to $1$ modulo $4$ implies that $t$ is even.
Let $\gamma: \R \rightarrow \R^{t+1}$ denote the function defined by $\gamma(x) = (1,x,x^2,\ldots,x^t)$.
For every $i \in [n]$, consider the point $w_i = \gamma(i) \in \R^{t+1}$.
It suffices to show that for every hyperplane $h$ in $\R^{t+1}$ that passes through the origin, at least one of the two open half-spaces that $h$ determines contains the points of $\{w_i \mid i \in A\}$ for some vertex $A$ of $U(n,k)$.

Consider an arbitrary hyperplane $h$ in $\R^{t+1}$ that passes through the origin.
Every point $w_i$ either lies on $h$ or belongs to one of the two open half-spaces determined by $h$.
Let $W_{on}$ denote the set of indices $i \in [n]$ for which the point $w_i$ lies on $h$, and let $W_1$ and $W_2$ denote the sets of indices $i \in [n]$ of the points $w_i$ that belong to the two open half-spaces determined by $h$. Our goal is to show that at least one of the sets $W_1$ and $W_2$ contains a vertex of $U(n,k)$.

The definition of the points $w_1,\ldots,w_n$ implies that the size of $W_{on}$ does not exceed the number of roots of some nonzero polynomial of degree at most $t$, hence
$|W_{on}| \leq t$, and it can be assumed that $|W_{on}|=t$.
The points $w_i$ with $i \in W_{on}$ divide the image $\Image(\gamma)$ of the function $\gamma$ into $t+1$ open continuous parts that alternate between the two open half-spaces determined by $h$. Since $t$ is even, the first and last parts lie on the same open half-space determined by $h$. We merge these two parts into a single part.
Observe that all the indices $i$ of the points $w_i$ that belong to each of the obtained $t$ parts of $\Image(\gamma)$ are either in $W_1$ or in $W_2$.
It follows that each of the sets $W_1$ and $W_2$ is associated with $t/2$ of these parts.

Suppose without loss of generality that $|W_1| \geq |W_2|$.
We claim that $W_1$ contains a vertex of $U(n,k)$. To this end, we show that it includes two consecutive elements modulo $n$ and has size at least $k$.
Indeed, it holds that $|W_1| > t/2$, as otherwise $n-t = |W_1|+|W_2| \leq t$, which implies that $n \leq 2 t$. Since $n$ is odd, this contradicts the definition of $t$ which guarantees that $t \leq \lceil n/2 \rceil -1$. It thus follows that at least one of the parts of $\Image(\gamma)$ associated with $W_1$ contains at least two of the points $w_1,\ldots, w_n$, hence $W_1$ includes two consecutive elements modulo $n$.
Further, by the definition of $t$, we have $t \leq n-2k+1$, which implies that $|W_1| \geq \lceil \frac{n-t}{2} \rceil \geq \lceil \frac{2k-1}{2} \rceil = k$, completing the proof.
\end{proof}

\begin{proof}[ of Theorem~\ref{thm:U(n,k)_n=1mod4}]
Let $n$ and $k$ be integers such that $n$ is congruent to $1$ modulo $4$ and $n \geq 2k$.
The upper bound on $\chi(U(n,k))$ follows from Corollary~\ref{cor:chiU(n,k)}.
For the lower bound, let \[t = \min (n-2k+2, \lceil n/2 \rceil) - 1,\]
and suppose for the sake of contradiction that there exists a proper coloring of $U(n,k)$ with $t$ colors. Let $y_1,\ldots,y_n \in \Sset^t$ denote the points given by Lemma~\ref{lemma:y_i_U_odd}. We define $t$ sets $F_1,\ldots,F_t \subseteq \Sset^t$ as follows. A point $x \in \Sset^t$ is included in $F_j$ with $j \in [t]$ if there exists a vertex $A$ of $U(n,k)$ colored $j$ such that $\{ y_i \mid i \in A\} \subseteq H(x)$, where $H(x)$ is the open hemisphere centered at $x$. We further define $F_{t+1} = \Sset^t \setminus (F_1 \cup \cdots \cup F_t)$. Note that the sets $F_1,\ldots,F_{t+1}$ cover $\Sset^t$.
By Theorem~\ref{thm:BU}, there exist an index $j \in [t+1]$ and a point $x \in \Sset^t$ such that both $x$ and $-x$ belong to $F_j$.
If $j \in [t]$, then it follows from the definition of $F_j$ that there exist two vertices of $U(n,k)$ with color $j$ that correspond to disjoint sets, contradicting the assumption that the given coloring is proper.
If $j = t+1$ then neither $H(x)$ nor $H(-x)$ contains $\{ y_i \mid i \in A\}$ for a vertex $A$ of $U(n,k)$, contradicting Lemma~\ref{lemma:y_i_U_odd}.
This completes the proof.
\end{proof}


\end{document}